\documentclass[twocolumn]{IEEEtran}
\usepackage{url,amsmath,graphicx}
\usepackage{amsmath,amssymb,amscd,latexsym,dsfont,amsthm}
 \usepackage{stfloats}
 \usepackage{eqparbox}
  \usepackage{array}
  \usepackage{fixltx2e}
  \usepackage{stfloats}
   \usepackage{cite}
   \usepackage{url}
\usepackage{cite}
\usepackage{graphicx}
\usepackage{psfrag}
\usepackage{url}
\usepackage{stfloats}
\usepackage{array}
\usepackage{fancyhdr}
\usepackage{psfrag}

\usepackage{amssymb}
\usepackage{color}

\newtheorem{proposition}{Proposition}
\newtheorem{remark}{Remark}

\begin{document}
	\title{Bit Error Rate Analysis for Reconfigurable Intelligent Surfaces  with Phase Errors}
\author{ Im\`ene~Trigui,~\IEEEmembership{Member,~IEEE}, Edouard Komi Agbogla, Mustapha Benjillali,~\IEEEmembership{Senior Member,~IEEE},\\ Wessam Ajib,~\IEEEmembership{Senior Member,~IEEE}, and Wei-Ping Zhu,~\IEEEmembership{Senior Member,~IEEE}
\thanks{I. Trigui, E. K. Agbogla, and  W. Ajib are with the D\'epartement d'informatique, Universit\'e du Qu\'ebec \`a Montr\'eal, Montreal, Canada. E-mails: trigui.imen@courrier.uqam.ca, agbogla.komi\_edouard@courrier.uqam.ca, and ajib.wessam@uqam.ca.}
\thanks{M. Benjillali is with INPT, Rabat, Morocco. Email: benjillali@ieee.org.}
\thanks{W.-P. Zhu is with the Department of Electrical and Computer Engineering, Concordia University, Montreal, Canada. E-mail: weiping@ece.concordia.ca.}
}

\maketitle
\begin{abstract}
In this paper, we analyze the error probability of reconfigurable intelligent surfaces (RIS)-enabled communication systems with quantized channel phase compensation over Rayleigh fading channels. The probability density and characteristic functions
of the received signal amplitude are derived and used to
compute  exact expressions for the bit error rate (BER). The  resulting  expressions  are  general,  as  they  hold  for  an  arbitrary number of reflecting elements $N$, and quantization levels, $L$. We introduce an exact asymptotic analysis in the high signal-to-noise ratio (SNR) regime, from which we demonstrate, in particular, that the diversity order is $N/2$ when $L=2$ and $N$ when $L>2$. The theoretical frameworks and findings are validated with the aid of Monte Carlo simulations.
\end{abstract}

\section{Introduction}
 Recently, a major design make-over of wireless communications enabled by reconfigurable intelligent surfaces (RISs), has become an  extremely active and promising research topic \cite{basar,ren4}.
Thanks to their capability to easily manipulate the incident wave proprieties (frequency, amplitude,  phase)  to establish favourable channel responses \cite{basar,ren4,qu} (and references therein), RIS offer undoubtedly an additional degree of freedom to boost capacity and enhance coverage with low cost and power-efficient infrastructure in 6G networks. Nevertheless, several hurdles must be overcome to enable RIS-assisted communications and make them work properly. One of the major challenges facing the implementation of RIS is their vulnerability to  phase compensation errors at RIS elements \cite{mal}. In fact, while a suitable design of the phase shifts of the reflecting
elements is necessary to reap the advantages of RIS-aided
transmission \cite{qu,mal, liang, hu}, it is, however,  difficult to implement these phases in practice due to hardware
limitations.  As a consequence, phase quantization errors inevitably arise. 

Motivated by these considerations,  recent  attempts for studying RIS-aided systems in the presence of phase errors include the use of approximate distributions and asymptotic analysis \cite{I1,M2,CL2,CL5,triguiris}. By using the central limit theorem (CLT), the authors of  \cite{I1} and \cite{M2} 
obtained  approximate BER expressions considering a large number of reconfigurable elements at the  RIS.  It was recently shown, however,  that the
approximation error attributed to the CLT can be significant for small number of elements with important discrepancies  in the  high signal-to-noise ratio (SNR) regime~\cite{I1,triguiris}.
  In \cite{CL2},  the authors derived an approximate error performance of RIS in the presence of phase errors by leveraging the moment-based Gamma approximation.
  The Gamma-based framework appears, however, unsuitable for diversity analysis since it fails to extract the full diversity order even in the absence of phase errors. Recently, the authors of  \cite{triguiris} and \cite{CL5} identified sufficient conditions based on upper and lower bounds for ensuring that RIS-assisted systems achieve the full diversity order in the presence of  phase  noise.  More specifically, it was shown in \cite{triguiris} that if  the  absolute  difference  between  pairs of  phase  errors  is  less  than $\pi/2$,  RIS-assisted  communications  achieve  full  diversity.

Although the results from \cite{I1,M2,CL2,CL5,triguiris} are insightful, these works have been successfully tractable
due to approximate SNR distributions and bounds yielding a diversity analysis in the presence of phase errors which is, so far, steadily inaccurate \cite{CL5, triguiris}. Moreover, to the best of our knowledge, no exact error analysis for RIS with quantized phase shifts and arbitrary number of reconfigurable elements  has been reported in the literature.

In this letter, we investigate the impact of quantized phase shifts on the performance of RIS-enabled communications. Analytical expressions for the bit error rate (BER) of binary phase shift keying (BPSK) modulated signals over Rayleigh fading channels, as a function of the number of
phase quantization levels $L$ (also known as phase  resolution) and the number of reflecting elements $N$, are derived based on
 exact expressions for the probability density function (PDF) and the characteristic function (CHF) of
the combined signal amplitude. Approximate expressions
for the BER are derived in the high-SNR regime yielding  
simple closed-form expressions which can be used to determine
the coding gain, the diversity order, and the degradation in the 
system performance. The obtained results unveil, in particular, that for two level quantization, the diversity order  is equal to  half  of  the  reconfigurable  elements number. For higher quantization levels, RIS-assisted  communications  achieve   full  diversity  order, despite the BER is not a linear function of the SNR (in dB scale).

\vspace{-0.2cm}
\section{System Model}
We consider  an RIS with $N$ reconfigurable elements which is deployed  to assist  data transmission to a single antenna receiver by reflecting an incident RF wave emitted by a single antenna transmitter. More specifically, we assume that the direct transmission link between the transmitter and the receiver is blocked, and, thus, the RIS is deployed to relay the scattered signal and to leverage virtual line-of-sight (LOS) paths for enhancing the strength of the received signal.  The received signal of the considered system is  \cite{qu}
\begin{equation}
r=\sum_{i=1}^{N} h_i g_i \rm{e}^{j (\theta_i-\theta_i^{{\cal Q}})}\alpha+ \widetilde{n},
\label{snoise1}
\end{equation}
where $\alpha$ is the baseband transmitted symbol, $\widetilde{n}$ represents a zero mean complex additive white
Gaussian noise (AWGN) with variance $N_0/2$, the channel
amplitudes $h_i$ and $g_i$ are assumed to be Rayleigh distributed with 
$\theta_i$ uniformly distributed between $[-\pi, \pi]$, and $\theta_i^{{\cal Q}}$ is denoting the quantized phase shift induced by the $i$-th reflecting element at the RIS. For an $L$-level uniform quantizer the signal phase can be
modeled as\vspace{-0.1cm}
\begin{equation}
 \theta_i^{{\cal Q}}=\theta_i- \theta_i^{{\cal E}},
\end{equation}
where $\theta_i^{{\cal E}}$
represents the phase error which is uniformly
distributed between $[-\pi/L,\pi/L]$.

In the rest of the work, we consider that $\alpha$ is chosen from a BPSK signal constellation set $\alpha \in \{-1, +1\}$. Hence, only the real
part of the signal has an impact on the detection performance\footnote{Although the analysis in this work is limited to BPSK
modulation for tractability reasons, the proposed method can be used to evaluate the BER of other modulation schemes as well. For instance, for QPSK modulation, the term $v_i$ in (\ref{snoise2}) can be replaced with 
 $v_i=\sqrt{2}\cos(\theta_i^{\cal E}\pm \frac{\pi}{4})$ \cite{ber}.
The detailed analysis of higher-order QAM modulations is more involved, and is an important extension that we will address in future work.}. This reduces \eqref{snoise1} to\vspace{-0.2cm}
\begin{equation}
r=\sum_{i=1}^{N} h_i g_i v_i \alpha+ {\rm Re}[\widetilde{n}],
\label{snoise2}
\end{equation}
where $v_i\triangleq \cos(\theta_i^{{\cal E}})$, and $\rm{Re}[\cdot]$ denotes the real part of a complex quantity.
The instantaneous SNR is defined as
\begin{equation}
\gamma=\rho\left(\sum_{i=1}^{N} h_i g_i v_i\right)^2,
\label{x}
\end{equation}
with $\rho$ denoting the transmit SNR.

\section{Average BER With Phase Noise}
The evaluation of the error probability  usually involves the computation of the PDF of $x=\sqrt{\gamma}$ which is formulated  in  terms   of  a  linear  combination of  the product of random  variables.   A common approach for analyzing the distribution of $x$ is to leverage the CLT. However, this approach is accurate only for a large number of reconfigurable elements \cite{basar, I1, M2} and the resulting analysis is usually not accurate in the high-SNR regime and, therefore, for analyzing the attainable diversity order \cite{triguiris}. In  general,  the calculation  of  the  exact  error probability of RIS-assisted communications with phase noise is  an  open  research  issue,  and  is  very intricate for arbitrary values of $N$.

To tackle this issue, we proceed by writing the BER expression using
the CHF-based approach as \cite{ana}
\begin{equation}
{\cal P}_e=\frac{1}{2\pi}\int_{-\infty}^{\infty}{\cal G}(t)\phi_x^{*}(t)\textrm{d}t,
\label{pe}
\end{equation}
where, under BPSK signalling, we have \cite{ana}
\begin{equation}
{\cal G}(t)=\frac{1}{2t}\left(\frac{t}{\sqrt{\pi}}{}_1{\rm{F}_1}\left( {1,\frac{3}{2};-\frac{t^{2}}{4}} \right)+j-j\rm{e}^{-\frac{t^{2}}{4}}\right),
\end{equation}
where $j=\sqrt{-1}$ and ${}_1{\rm{F}_1}(\cdot)$ represents the confluent hypergeometric
function\cite{grad}.
Moreover, we denote $\phi_x^{*}(t)$ the complex conjugate of the CHF of $x$ defined as\vspace{-0.2cm}
\begin{equation}
\phi_x(t)=\prod_{i=1}^{N} \phi_{z_i}\left(\sqrt{\rho}t\right),
\label{px}
\end{equation}
where $z_i=h_i g_i v_i$, $i=1,\ldots,N$, and $\phi_{z_i}\left(t\right)=\mathbb{E}_{z_i}\{\rm{e}^{j t z_i }\}$, with $\mathbb{E}\{\cdot\}$ denoting the expectation operator.
 Thus, to evaluate the error
probability, we need first to determine the CHF of $z_i$, which is a product of three random variables.

\begin{proposition}
The PDF of $z_i$ can be expressed as\vspace{-0.17cm}
\begin{equation}
f_{z_i}(z)=L {\rm{e}}^{-2z}-\frac{2 L}{\pi} \int_{0}^{\pi/2} \rm{e}^{-2z\left(\sqrt{1+\frac{\tan\left(\pi/L\right)^{2}}{\sin(\psi)^{2}}}\right)} {\rm{d}}\psi.
\label{p1}
\end{equation}
\end{proposition}

\begin{proof}
As noted earlier,  $\theta_i^{{\cal E}}$ is uniformly distributed on the interval
$[-\frac{\pi}{L}, \frac{\pi}{L}]$. As a result, the PDF of $v_i$ is given by
\begin{equation}
f_{v_i}(v)=\frac{L}{\pi\sqrt{1-v^{2}}}, \quad\text{for~} \cos(\pi/L)\leq v \leq1.
\label{cs}
\end{equation}
Recall that under independent and identically distributed (i.i.d.) Rayleigh fading $f_{y_i}(x)= 2 x {\rm{e}}^{-x^{2}}$, for $y\in \{h,g\}$, then we use (\ref{cs}) to calculate the PDF  of $z_i=h_i g_i v_i$ as
\begin{equation}
f_{z_i}(z)=\int_{0}^{\infty}\frac{1}{x}f_{h_i}(x)f_{g_i v_i}\left(\frac{z}{x}\right) \textrm{d}x,
\label{eqf}
\end{equation}
where\vspace{-0.22cm}
\begin{equation}
f_{g_i v_i}\left(x\right)=\frac{2L}{\pi}\int_{\cos({\pi/L})}^{1}\frac{\rm{e}^{-\frac{x^{2}}{z^{2}}}}{z^{2}\sqrt{1-z^{2}}}\textrm{d}z.
\label{eq1}
\end{equation}
By relabeling $z= \sin(\theta)$ in (\ref{eq1}), and applying \cite[Eq. (2.33.2)]{grad}, we obtain
\begin{equation}
f_{g_i v_i}\left(x\right)=\frac{L}{\sqrt{\pi}}\rm{e}^{-x^{2}}{\rm erf}\!\left(x \tan\left(\frac{\pi}{L}\right)\right),
\label{eq2}
\end{equation}
 where ${\rm erf}(\cdot)$ is the error function \cite{grad}. Substituting (\ref{eq2}) into (\ref{eqf}) while resorting to the alternative form
 \begin{equation}
  {\rm erf}(z)=1-\frac{2}{\pi}\int_{0}^{\pi/2} \rm{e}^{-\frac{z^{2}}{\sin(\psi)^{2}}} \textrm{d}\psi,
  \end{equation}
we obtain the PDF in (\ref{p1}) with the aid of \cite{grad}. This completes the proof.
\end{proof}

\begin{remark}
Based on Proposition 1, we evince that for 1-bit quantization based RIS communications (i.e., when $L=2$), the PDF of $z_i$ reduces to
\begin{equation}
f_{z_i}(z)=2 \rm{e}^{-2z}, \quad z\geq 0
\label{eqf1}
\end{equation}
which is in agreement with \cite[Eq. (17)]{M2}.
\end{remark}

Using (\ref{p1}), the CHF of $z_i$  follows  from $\phi_{z_i}(t)=\int_{0}^{\infty} f_{z_i}(z){\rm{e}}^{j t z} \textrm{d}z$  and is obtained as
\begin{align}
&\phi_{z_i}(t)=\frac{L}{jt+ 2}\nonumber \\ 
&\quad-\frac{2 L}{\pi}\int_{0}^{\pi/2} \left(jt+2\left(\sqrt{1+\frac{\tan\left(\frac{\pi}{L}\right)^{2}}{\sin(\psi)^{2}}}\right) \right)^{-1} \textrm{d}\psi.
\label{p2}
\end{align}
Note that the CHF given in (\ref{p2}) can be efficiently estimated using
the Gauss-Chebyshev quadrature (GCQ) \cite[Eq.(25.4.38)]{abr} as
\begin{eqnarray}
\!\!\!\!\!\phi_{z_i}(t)\!\!\!&\approx&\!\!\!\frac{L}{jt+ 2}\!-\!\frac{L\pi}{2n}\sum_{k=0}^{n} \frac{\sqrt{1-a^{2}_k}}{jt+2\sqrt{1\!+\!\frac{\tan\left(\frac{\pi}{L}\right)^{2}}{\sin\left(\frac{\pi}{4}(a_k+1)\right)^{2}}}},
\label{p3}
\end{eqnarray}
where $a_k=\cos\left(\frac{\pi}{2n}(2k-1)\right)$. It is important to note that the accuracy of the  GCQ
rule is extremely high; and a  relative accuracy of $10^{-15}$ is possible
for all SNRs and  RIS configurations (i.e., $N$, $L$).  Thus  using (\ref{p3}) and (\ref{px}) together
with (\ref{pe}) can be considered as a replacement
for the closed-form solution to the BER.  In particular, using GCQ leads to
significant computational advantages as compared to the case
when (\ref{pe}) is used with (\ref{p2}).

\section{High-SNR BER-Diversity Analysis}

In this section, we will analyze the communication robustness that can be achieved with
$L$-level quantization-based RIS by focusing on the BER at high SNR. The objective is to quantify the diversity order and identify sufficient conditions for achieving it.   This  is  a  fundamental  open  issue  for  designing  and  optimizing  RIS-aided  systems. So far, by relying on the CLT, it was shown in  \cite{bour}, for example, that the diversity order is $\frac{N}{2}\frac{\pi^{2}}{16-\pi^{2}}$ in Rayleigh fading, which implies that the full diversity order cannot be obtained even in the absence of phase errors. By resorting to some bounds, however, the authors of \cite{triguiris} recently showed that the full diversity order is achievable in Rayleigh fading if  the  absolute  difference  between  pairs of  phase  errors  is  less  than $\pi/2$.

In  what  follows,  building  upon  the  exact high-SNR  analysis  of  BER for arbitrary $N$, we compute the exact diversity order and coding gain of RIS-assisted  systems. In  order  to  evaluate  the  BER at high SNR, according  to  \cite{gian},  we  are  interested  in  analyzing the behavior of
the PDF of $x$ around the origin. Using (\ref{eq1}) and  resorting to the Mellin-Barnes integral representation of the exponential and error functions \cite{mathai}, it follows that
\begin{align}
&f_{z_i}(z)=\frac{4 L}{ \pi (2\pi j)^{2}}\int_{{\cal C}_1} \int_{{\cal C}_2} \frac{\Gamma(s_1)\Gamma(s_2+\frac{1}{2})\Gamma(-s_2)}{\Gamma(1-s_2)} \nonumber \\ 
&\quad\int_{0}^{\infty} {\rm{e}}^{-x^{2}} \left(\frac{z}{x}\right)^{-2 s_1}\left(\frac{z\tan\left(\frac{\pi}{L}\right)}{x}\right)^{-2 s_2} \textrm{d}x \textrm{d}s_1 \textrm{d}s_2.
\label{eq3}
\end{align}
where ${\cal C}_1$ and ${\cal C}_2$ are two integral contours in the complex
domain. Then, with the help of $\int_{0}^{\infty} x^{2b }\rm{e}^{-x^{2}}\textrm{d}x=\Gamma\left(\frac{1}{2}+b\right)/2$ and applying \cite[Eq. (2.1)]{mathai}, we obtain
\begin{align}
f_{z_i}(z)&\overset{(a)}{=}\frac{2 L}{\pi} \nonumber \\
&\hskip-0.8cm\times {\rm H}_{0,1:0,1; 1, 1}^{0,0:1, 0; 1, 2}\left[\!\!\begin{array}{ccc}z^{2} \\ z^{2}\tan\left(\frac{\pi}{L}\right)^{2}\end{array}\!\! \Bigg|  \!\!\begin{array}{ccc} -: -;(1,1)\\ (\frac{1}{2};1,1)\!:\!\left(0,1\right);   \left(\frac{1}{2},1\right), (0,1) \end{array}\Bigg. \!\!\right]\nonumber \\
&\underset{z\rightarrow0}{\overset{(b)}{\approx}}\frac{4 L}{ \pi }z \ln\left(z\tan\left(\frac{\pi}{L}\right)\right)\tan\left(\frac{\pi}{L}\right),
\label{eq4}
\end{align}
where $(a)$ follows from  applying \cite[Eq. (2.57)]{mathai}, and $H[\cdot|  \cdot]$ is the bivariate Fox’s H-function \cite[ Eq. (2.56)]{mathai}. When $z\rightarrow0$, $(b)$ follows by computing  residues at left poles of the two corresponding integrands \cite[Theorem 1.3]{kilbas}.\\
The Laplace Transform  of the approximated $f_{z_i}(z)$ (i.e., when $z\to0$) is
now given by
\begin{equation}
{\cal L}_{z_i}(s)=\frac{2 L \tan\left(\frac{\pi}{L}\right)}{ \pi} \frac{\ln\left(s\right)}{s^{2}}.
\end{equation}
Accordingly, the Laplace transform of the approximated PDF of  $x=\sqrt{\rho}\sum_{i=1}^{N}z_i$ can be formulated as
\begin{equation}
{\cal L}_{x}(s)=\left(\frac{2L \tan\left(\frac{\pi}{L}\right)}{ \pi } \frac{\ln\left(\sqrt{\rho}s\right)}{\rho s^{2}}\right)^{N}.
\label{l1}
\end{equation}
The PDF of $x$ requires the computation of the inverse Laplace transform of (\ref{l1}) which is involved due to the term  $\ln(\sqrt{\rho} s)^{N}$. To proceed with the analysis, we use the following identity
\begin{align}
    \ln(\sqrt{\rho} s)^{N}=&\frac{1}{(2\pi j)^{N}}\int_{{\cal C}_1}\ldots\int_{{\cal C}_N}\frac{\prod_{i=1}^{N}\Gamma(1+t_i) \Gamma(-t_i)^{2}}{\prod_{i=1}^{N}\Gamma(1-t_i)} \nonumber \\ 
    & \quad\left(\sqrt{\rho} s-1\right)^{-\sum_{i=1}^{N}t_i}\textrm{d}t_1\ldots \textrm{d}t_N.
    \label{l4}
\end{align}
Then we plug the above equality into (\ref{l1}) and we compute the inverse Laplace transform with respect to $s$ by using \cite[Eq. (2.21)]{mathai} as
\begin{align}
\!\!\!{\cal L}^{-1}\left\{\left(1\!- a s^{-1}\right)^{-c} s^{-c-2N},x\right\}
&\underset{x\rightarrow 0}{\overset{(b)}{=}}&\!\!\!\frac{ a^{-c} x^{c+2N-1}}{\Gamma(2N+c)},
\label{l2}
\end{align}
where $(b)$ holds since we are interested in values around zero.
Hence, the PDF of $x$ around zero is approximated using (\ref{l1}), (\ref{l4}) and (\ref{l2}), and can be formulated as
\begin{align}
f_x(x)&=\frac{\left(\frac{2 L\tan\left(\frac{\pi}{L}\right)}{ \pi }\right)^{N}x^{2N-1}}{(2\pi j)^{N}}\nonumber \\
& \times\int_{{\cal C}_1}\ldots\int_{{\cal C}_N}\frac{\prod_{i=1}^{N}\Gamma(1+t_i) \Gamma(-t_i)^{2}}{\Gamma(2N+\sum_{i=1}^{N}t_i)\prod_{i=1}^{N}\Gamma(1-t_i)} \nonumber \\
& \times x^{-\sum_{i=1}^{N}t_i}\textrm{d}t_1\ldots \textrm{d}t_N.
\label{l3}
\end{align}
which  can  be  expressed  in  terms  of  the  multivariate  Fox's  H-function \cite[A.1]{mathai} whose details are not provided due to space
limitation. When $x\rightarrow 0$, we apply the asymptotic expansion of the Mellin-Barnes integrals in (\ref{l3}) at  the double poles $t_i=0$ using \cite[Eqs. (1.8.14), (1.4.6)]{kilbas}, thereby yielding
\begin{equation}
f_x(x)\underset{x\rightarrow0}{\approx} \left(\frac{2 L\tan\left(\frac{\pi}{L}\right)}{ \rho \pi }\right)^{N}\frac{x^{2N-1} \ln(\rho x)^{N}}{\Gamma(2N)}.
\label{fx}
\end{equation}
Using (\ref{fx}), we can now  formulate the asymptotic BER as \cite{ana}
\begin{equation}
{\cal P}_\text{e}=\frac{\left(\frac{2 L \tan\left(\pi\right)}{ \rho \pi }\right)^{N}}{2\Gamma(2N)}\int_{0}^{\infty}{\rm erfc}(\sqrt{x})x^{N-1} \ln(\rho \sqrt{x})^{N}\textrm{d}x,
\label{pe3}
\end{equation}
where ${\rm erfc}(x)$ is the complementary error function.
By rewriting the
integrands in (\ref{pe3}) as Fox's H-functions using \cite[Eq. (1.43)]{mathai} and applying \cite[Eq. (2.3)]{mathai}, we get (\ref{pe4}) at the top of the page, where ${\rm H}[\cdot]$ stands for the Fox's H-function \cite[Eq. (1.2)]{mathai}. The asymptotic expression
of the BER when $\rho\to \infty$ can be  obtained
by representing the expression in (\ref{pe4}) as a multiple
Mellin-Barnes type integral and then computing the residues of
the integrands using \cite[Eq. (1.5.12)]{kilbas}.
\begin{figure*}
\begin{eqnarray}
{\cal P}_\text{e}&=&\frac{\left(\frac{2L\tan\left(\frac {\pi}{L}\right)}{ \rho \pi }\right)^{N}}{2\sqrt{\pi}\Gamma(2N)}
\frac{1}{(2\pi j)^{N}}\int_{{\cal C}_1}\ldots\int_{{\cal C}_N}\!\!\!\frac{\prod_{i=1}^{N}\Gamma(1+t_i) \Gamma(-t_i)^{2}}{\prod_{i=1}^{N}\Gamma(1-t_i)} \nonumber \\ && \times ~\frac{(-1)^{-\sum_{i=1}^{N}t_i}}{\Gamma(\sum_{i=1}^{N}t_i)} H_{3,2}^{1,3}\Bigg[ -\rho \left|\begin{array}{ccc} (1-N,\frac{1}{2}), (\frac{1}{2}-N,\frac{1}{2}), (1-\sum_{i=1}^{k}t_i,1) \\ (0,1),(\frac{1}{2},1) (-N,1)
\end{array}\right. \Bigg] \textrm{d}t_1\ldots \textrm{d}t_N.
\label{pe4}
\end{eqnarray} \hrulefill
\end{figure*}
Hence from (\ref{pe4}), we can finally compute the asymptotic BER as
\begin{equation}
{\cal P}_\text{e}=\left(\frac{2 L\tan\left(\frac{\pi}{L}\right)}{\pi }\right)^{N}\frac{\ln(\rho)^{N}\Gamma(N+\frac{1}{2})}{2\sqrt{\pi}(N+1)\Gamma(2N)} \rho^{-N},
\label{pe2}
\end{equation}
 from which the following important conclusions and performance trends are unveiled.
 \begin{itemize}
 \item Due to the presence of $\tan(\frac{\pi}{L})$, (\ref{pe2}) holds for $L>2$.
 \item (\ref{pe2})  unveils a new scaling  law of the BER which is $\ln(\rho)/\rho$ as $\rho\to \infty$.  The
BER is not a linear function with respect to $\rho$ in dB, while
the slope of BER changes very slowly at high SNR. 
 This new scaling law  generalizes  the  definitions of diversity order and coding gain typically used in wireless communications \cite{gian}.
     \item (\ref{pe2}) is, to the best of our knowledge, the first in the literature that yields the exact asymptotic BER  for arbitrary $N$. This is in contrast with
the recently reported expressions in \cite[Eqs. (4), (7)]{basar}, \cite[Eq. (29)]{I1} and \cite[Eq. (39)]{CL2}, which are based on approximations (the CLT in \cite{basar}, \cite{I1} and the moment-based Gamma approximation in \cite{CL2}).
\end{itemize}

In order to better quantify the diversity order of RIS-assisted systems in the presence of quantized  phase noise, it is important to study the case when $L=2$. To this end,  the asymptotic PDF of the normalized SNR $\gamma/\rho$ follows by resorting to the fact that ${\cal L}_{\frac{x}{\sqrt{\rho}}}(s)=\left(2{\cal L}\{\rm{e}^{-2z}, s\}\right)^{N}=\left(\frac{s}{2}+1\right)^{-N}$, and hence
\begin{equation}
f_{\frac{\gamma}{\rho}}(\gamma)\approx \frac{2^{\frac{N}{2}-1}\gamma^{\frac{N}{2}-1}}{\Gamma(N)}\rm{e}^{-2\sqrt{\gamma} },
\end{equation}
which implies, according to \cite{gian}, that the asymptotic error probability is
\begin{equation}
{\cal P}_e\underset{\rho\to \infty}{\approx}\frac{2^{N-1}\Gamma\left(\frac{N}{2}+\frac{1}{2}\right)}{\sqrt{\pi}\Gamma(N+1)}\rho^{-\frac{N}{2}}.
\label{pe5}
\end{equation}
Given (\ref{pe2}) and (\ref{pe5}),  the  diversity and coding gains of RIS-assisted communications with BPSK signalling and quantized phase noise are obtained, as stated in the following proposition.

\begin{proposition}
The diversity order and coding gains of RIS-assisted communications  in the presence of phase noise are
\begin{equation}
{\cal G}_d=\left\{
                                                              \begin{array}{ll}
                                                                N/2, & \hbox{ $ L=2$;} \\
                                                                N, & \hbox{$L>2$.}
                                                              \end{array}
                                                            \right.
\end{equation}
and\vspace{-0.15cm}
\begin{equation}
{\cal G}_c=\begin{cases}
\!\!\left(\dfrac{2^{N-1}\Gamma(\frac{N}{2}+\frac{1}{2})}{\sqrt{\pi}\Gamma(N+1)}\right)^{\!\!-2/N}, & L=2;\\
\!\!\left(\!\!\left(\dfrac{2 L \tan\left(\frac{\pi}{L}\right)}{\pi }\!\!\right)^{\!N}\!\!\dfrac{\ln(\rho)^{N}\Gamma(N+\frac{1}{2})}{2\sqrt{\pi}(N+1)\Gamma(2N)}\!\!\right)^{\!\!-1/N}, & L>2.
\end{cases}
\end{equation}
\end{proposition}

The diversity analysis above helps to discover the effects of an $L$-level\footnote{Considering the fact that energy consumption at RIS increases exponentially with  the number of levels $L$ (also known as quantizer resolution)\cite{res}, low-resolution quantizers, i.e., with small $L$ values, could provide significant energy savings, notably, when the number of elements $N$ at RIS is large.} quantizer
based RIS on the BER system performance. In particular, we conclude that
\begin{itemize}

\item For RIS-aided systems with BPSK signalling over Rayleigh fading,  full diversity order can be ensured if at least three quantization levels (i.e. $L>2$)  are used. This result is in agreement with \cite{CL5, triguiris} where it was proved by using upper and lower bounds.

    \item For $L=2$ the diversity order is only $N/2$.  This is consistent with \cite{CL5}, where it is proved that the diversity order cannot exceeds $(N+1)/2$ for $L=2$. To the   best  of  the  authors  knowledge,  this  difference  in  the  diversity gain between the scenarios $L=2$ and $L>2$ was never reported in the literature.
        \item The increase in the average BER as a quantization penalty is defined as
\begin{align}
&\Psi(\rho, L)=10\log\!\left(\frac{{\cal P}_e}{{\cal P}^{\infty}_e}\right)\nonumber \\ &=\begin{cases}10\log\!\left(\!\dfrac{\rho^{\frac{N}{2}}}{\ln(\rho)^{N}}\dfrac{\Gamma(\frac{N}{2}+\frac{1}{2})(N+1)2^{2N-1}}{N\sqrt{\pi}}\!\right)\!, & \!\!\!L=2; \\
10\log\!\left(\!\dfrac{L\tan(\pi/L)}{\pi}\!\right)\!, & \!\!\!L>2,
\end{cases}
\end{align}
where ${\cal P}^{\infty}_e$ is the average SEP with infinite number of quantization bits obtained form (\ref{pe2})
by using the small-angle approximation $\tan(x)\approx x$ as $x\to 0$.
\end{itemize}

\section{Numerical Results}
In this section, we present analytical and simulated BER results for RIS-assisted communications using an $L$-level quantized phase shifts. The CGQ-based CHF in (\ref{p3}) is used to numerically evaluate the performance of BPSK
signalling over Rayleigh fading channels.

Fig.~1 shows the average BER  as a function of the average transmit SNR (i.e., $\rho$) with $L= 2, 3, 4$ and $N=5$. We observe that the analytical expression of  the  BER  using (\ref{p3}) and  its  corresponding  high-SNR  approximation  in  (\ref{pe2}) and (\ref{pe5}) are  in close agreement with Monte Carlo simulations. In particular, with GCQ, a high level of
accuracy is achieved with only $n=20$. We observe a noteworthy improvement in the average BER when
$L$ changes from 2 to 3. This increase in the average BER performance is expected in the light of Proposition 2, which states that the full diversity order can be achieved as long as $L\geq 3$. When $L=2$, however, the quantization errors induce a diversity gain loss. 
\vspace{-0.2cm}
\begin{figure}[h!]
\psfrag{mmmmmmmmmmmmm}[l][l][0.75][0]{Analytical, $L=2$}
\psfrag{jjjjjjjjjjjjjj}[l][l][0.75][0]{Analytical, $L=3$}
\psfrag{mmmmmmmmmmmmm}[l][l][0.75][0]{Analytical, $L=2$}
\psfrag{jjjjjjjjjjjjjj}[l][l][0.75][0]{Analytical, $L=3$}
\centering
\includegraphics[scale=0.37]{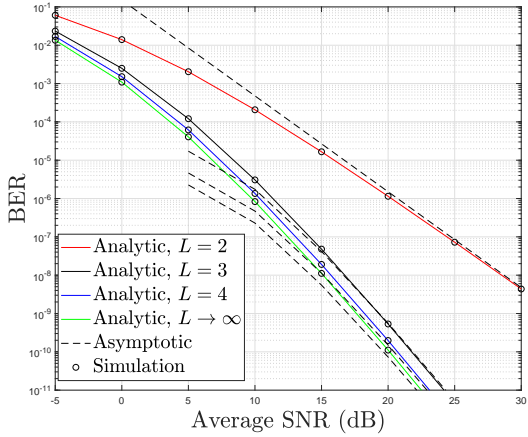}
\caption{ BER versus average SNR for $N=5$.}
\end{figure}

Fig.~2 shows a  comparison between the exact and approximate
BER expressions for several values of $N$ and $L$.   We observe, in particular, that the BER scales with ${{\ln \left( \rho  \right)} \mathord{\left/ {\vphantom {{\ln \left( \rho  \right)} \rho }} \right. \kern-\nulldelimiterspace} \rho }$, as predicted by (\ref{pe4}) and unveiled in \cite{triguiris} over  Rayleigh fading channels and in the absence of phase noise. As expected, the BER decreases significantly as the number $N$ of reconfigurable elements of the RIS increases.

\begin{figure}[h!]
\centering
\includegraphics[scale=0.37]{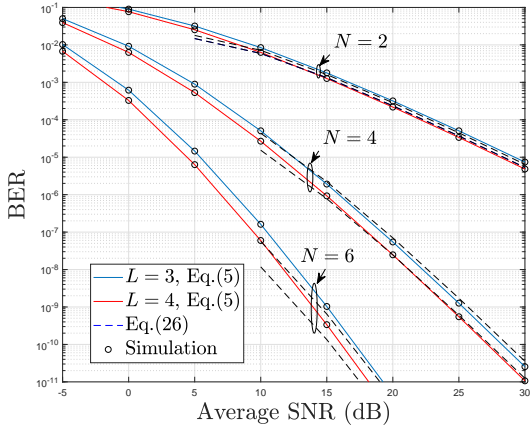}
\caption{BER versus average SNR, for different values of $L$ and $N$.}
\end{figure}
\vspace{-0.2cm}
\section{Conclusion}
This letter investigates the effect of low resolution quantization based-RIS  on the BER system performance. 
By conducting asymptotic analysis, this letter concisely unveils  the diversity order, coding gain and the
degradation in system performance with respect to quantization levels. It was shown, in particular, that  for very low resolution quantizers (i.e. $L=2$), RIS-assisted communications achieve  a diversity order equal to  half of the reconfigurable elements number. For higher values of $L$, however, a full diversity is extracted, despite the BER slowly changing slope in high
SNRs.


\begin{thebibliography}{9}



\bibitem{basar} E.  Basar,  M.  Di  Renzo,  J.  de  Rosny,  M.  Debbah,  M.-S.  Alouini,  and  R.  Zhang,  "Wireless  communications  through reconfigurable intelligent surfaces," \textit{IEEE Access}, vol. 7, pp. 116753-116773, Aug. 2019.

\bibitem{ren4} M. Di Renzo \textit{et al.}, ``Smart radio environments empowered by reconfigurable intelligent surfaces: How it works, state of research, and the road ahead," \textit{IEEE J. Sel. Areas Commun.}, vol. 38, no. 11, pp. 2450-2525, Nov. 2020.

\bibitem{qu} Q. Wu and R. Zhang, "Towards smart and reconfigurable environment: Intelligent reflecting surface aided wireless
network," \textit{IEEE Commun. Mag.}, vol. 58, no. 1, pp. 106-112, Jan. 2020.

\bibitem{mal} J. Ye, S. Guo, and M.-S. Alouini, "Joint reflecting and precoding designs
for {SER} minimization in reconfigurable intelligent surfaces assisted {MIMO}
systems," \textit{IEEE Trans. on Wireless Commun.}, vol. 19,
no. 8, pp. 5561-5574, Aug. 2020.

\bibitem{liang} L.  Yang,  Y.  Yang,  D.  B.  da  Costa,  and  I.  Trigui,  “Outage  probability  and  capacity  scaling  law  of  multiple  RIS-aided cooperative networks”, 2020. [Online]. Available: arXiv:2007.13293

\bibitem{hu} C. Huang, A. Zappone, G. C. Alexandropoulos, M. Debbah, and C. Yuen, "Reconfigurable intelligent surfaces for energy efficiency in wireless communication," \textit{IEEE Trans. Wireless Commun.}, vol. 18, no. 8, pp. 41574170, Aug. 2019.




\bibitem{I1} R. C. Ferreira, M. S. P. Facina, F. A. P. De Figueiredo, G. Fraidenraich,
and E. R. De Lima, "Bit error probability for large intelligent surfaces
under double-{N}akagami fading channels," \textit{IEEE Open Journal of the
Communications Society,} vol. 1, pp. 750-759, May 2020.

\bibitem{M2} M. Badiu and J. P. Coon, ``Communication through a large reflecting
surface with phase errors," \textit{IEEE Wireless Commun. Lett.}, vol. 9, no. 2,
pp. 184-188, Feb. 2020.


\bibitem{triguiris} I. Trigui, W. Ajib, W.-P. Zhu, M. Di Renzo, "Performance evaluation and diversity analysis of {RIS}-assisted communications over generalized fading channels in the presence of phase noise". [Online]. Available: arxiv.org/abs/2011.12260.

\bibitem{CL2} F. A. P. De Figueiredo, M. S. P. Facina, R. C. Ferreira, and G. Fraidenraich, "Large intelligent surfaces with discrete set of phase-shifts
communicating through double-{R}ayleigh fading channels". [Online]. Available: \textrm{d}x.doi.org/10.36227/techrxiv.13040654.v1.


 \bibitem{CL5} P. Xu, G. Chen, Z. Yang, and M. Di Renzo, ``Reconfigurable intelligent surfaces assisted communications with discrete phase shifts: How many quantization levels are required to achieve full diversity ?", 2020. [Online]. Available: arXiv:2008.05317.
\bibitem{M1} T. Wang, G. Chen, J. P. Coon, and M.-A. Badiu, ``Study of intelligent reflective surface assisted communications with one-bit phase adjustments,'' 2020. [Online]. Available: arXiv:2008.09770.


\bibitem{ber} M. A. Najib and V. K. Prabhu, ``Analysis of equal-gain diversity with
partially coherent fading signals,'' \textit{IEEE Trans. Veh. Technol.,} vol. 49, no. 5, pp. 783–791, May 2000.

\bibitem{ana} A. Annamalai, C. Tellambura, and V. K. Bhargava, "Equal-gain diversity
receiver performance in wireless channels," \textit{IEEE Trans. Commun.,}
vol. 48, no. 10, pp. 1732-1745, Oct. 2000.


\bibitem{bour} A.-A. A. Boulogeorgos and A. Alexiou, ``Performance analysis of
reconfigurable intelligent surface-assisted wireless systems and comparison with relaying," \textit{IEEE Access,} vol. 8, pp. 94463-94483, 2020.




\bibitem{gian} Z. Wang and G. B. Giannakis, ``A simple and general parameterization quantifying performance in fading channels,'' \textit{IEEE Trans. Commun.}, vol. 51, no. 8, pp. 1389-1398, Aug. 2003.

\bibitem{mathai} A. M. Mathai, R. K. Saxena, and H. J. Haubold, \textit{The {H}-Function: Theory and Applications}, Springer Science \&  Business Media, 2009.

\bibitem{kilbas} A. Kilbas and M. Saigo, \textit{H-Transforms: Theory and Applications}, CRC Press, 2004.

 \bibitem{grad} I. Gradshteyn and I. Ryzhik, \textit{Table of Integrals, Series, and Products}, Academic Press, 1994.
 
\bibitem{abr} M. Abramowitz and I. A. Stegun, \textit{Handbook of Mathematical Functions
with Formulas, Graphs and Mathematical Tables}, 9th edition. Dover
Publications, 1970.
\bibitem{res} J. Zhang, L. Dai, X. Li, Y. Liu, and L. Hanzo, ``On low-resolution ADCs in practical 5G millimeter-wave massive MIMO systems,'' \textit{IEEE Commun. Mag.}, vol. 56, no. 7, pp. 205-211, Jul. 
\end{thebibliography}
\end{document}